\documentclass[pra,10pt,twocolumn,superscriptaddress,floatfix,showpacs]{revtex4-1}

\usepackage[utf8]{inputenc}  
\usepackage[T1]{fontenc}     
\usepackage[british]{babel}  
\usepackage[sc,osf]{mathpazo}
\usepackage[scaled=0.86]{berasans}  
\usepackage[colorlinks=true, citecolor=blue, urlcolor=blue]{hyperref}  
\usepackage{graphicx} 
\usepackage[babel]{microtype}  
\usepackage{amsmath,amssymb,amsthm,bm,amsfonts,mathrsfs,bbm} 

\usepackage{xspace}  
\usepackage{pgfplots}
\usepackage{xcolor,colortbl}
\def\ba{\begin{equation}}
\def\ea{\end{equation}}
\def\bea{\begin{eqnarray}}
\def\eea{\end{eqnarray}}
\def\ben{\begin{equation*}}
\def\een{\end{equation*}}
\def\bean{\begin{eqnarray*}}
\def\eean{\end{eqnarray*}}
\def\bma{\begin{mathletters}}
\def\ema{\end{mathletters}}
\def\bi{\begin{itemize}}
\def\ei{\end{itemize}}

\newcommand{\be}{\begin{equation}}
\newcommand{\ee}{\end{equation}}

\newcommand{\kommentar}[1]{}

\newcommand{\forget}[1]{}

\newtheorem{theorem}{Theorem}

\newtheorem{condition}[theorem]{Condition}
\newtheorem{conjecture}[theorem]{Conjecture}
\newtheorem{corollary}[theorem]{Corollary}

\newtheorem{definition}[theorem]{Definition}
\newtheorem{example}[theorem]{Example}

\newtheorem{lemma}[theorem]{Lemma}

\newtheorem{observation}[theorem]{Observation}

\begin{document}

\title{Persistency of Genuine Correlations Under Particle Loss}


\author{Biswajit Paul}
\affiliation{Department of Mathematics, South Malda College, Malda, West Bengal, India}

\author{Kaushiki Mukherjee}
\affiliation{Department of Mathematics, Goverment Girls' General Degree College, Ekbalpore, Kolkata-700023, India}

\author{Ajoy Sen}
\affiliation{Department of Applied Mathematics, University of Calcutta, 92, A.P.C Road, Kolkata-700009, India}

\author{Debasis Sarkar}
\affiliation{Department of Applied Mathematics, University of Calcutta, 92, A.P.C Road, Kolkata-700009, India}

\author{Amit Mukherjee}
\affiliation{Physics and Applied Mathematics Unit, Indian Statistical Institute, 203, B.T. Road, Kolkata-700108, India}

\author{Arup Roy}
\affiliation{Physics and Applied Mathematics Unit, Indian Statistical Institute, 203, B.T. Road, Kolkata-700108, India}

\author{Some Sankar Bhattacharya}
\affiliation{Physics and Applied Mathematics Unit, Indian Statistical Institute, 203, B.T. Road, Kolkata-700108, India}


\begin{abstract}
In a recent work [\href{http://journals.aps.org/pra/abstract/10.1103/PhysRevA.86.042113}{Phys. Rev. A 86, 042113 (2012)}] the question of persistency of entanglement and nonlocality of multi-party systems under particle loss has been addressed. This question is of immense importance considering the practical realization of the information theoretic tasks which make use of the power of quantum correlations. But in multipartite scenario more interesting cases arise since subsystems can also possess genuineness in correlation which is prevalently inequivalent to the bipartite scenario. In this work, we investigate the persistency of such genuine correlations under particle loss. Keeping in mind the practical importance, considerable attention has been devoted to find the multipartite states which exhibit maximal persistency of genuine correlations.       
\end{abstract}


\maketitle

	
\section{Introduction}

Correlations play a fundamental role in quantum information theory. Entanglement \cite{Reviw Ent}, quantum steering \cite{Wiseman07} and nonlocality \cite{Reviw NL} are considered as prime features of quantum correlations. Quantum entanglement is a physical phenomenon that occurs when many particles are generated or interacted in such a fashion that the quantum state of each particle cannot describe the full system separately but it can be described holistically only. This kind of quantum states can be used to demonstrate nonlocality where the statistics generated from each subsystems can not be reproduced by any local realistic theory \cite{Bell64}. Bell nonlocal correlation along with entanglement are found to be the key resource for many information processing tasks such as teleportation \cite{Bennett93}, dense coding \cite{Bennett92}, randomness certification \cite{random}, key distribution \cite{key}, dimension witness\cite{dw}, Bayesian game theoretic applications \cite{game} so on so forth. Quantum steering \cite{Wiseman07} is a scenario where one party can remotely prepare the state of other party who are spatially separated, by applying suitable choice of measurements. This curious feature has also found applications in one-sided device independent cryptography\cite{Branciard12}. In a recent work\cite{Quintino15}, it has been shown that entanglement, steering and nonlocality are inequivalent notions under general quantum operations. 

Entanglement, steering and nonlocality are well understood in the bipartite scenario. But for more than two parties complexity increases resulting multipartite case to be richer in essence\cite{Reviw Ent}. In recent past a considerable number of attempts have been made to understand the genuine multipartite correlations which are remarkably different from their bipartite counterpart. An n-partite entangled state will be called genuine\cite{gme} if and only if the state is not separable with respect to any m-partition ($m \le n$) of the subsystems. Being useful resource for computation\cite{comp}, simulation\cite{sim}, metrology\cite{metr}, study of genuine multi-partite entanglement is a field of latest attraction. It is even useful for dinning cryptography problem \cite{Ramij15}. Similarly multipartite nonlocality is also not very easy to understand compared to the bipartite cases and found to be an important resource. Extension of quantum steering phenomena to multipartite scenario has been recently explored \cite{multisteer}. Even genuineness of steering found some importance in a number of recent works. So any kind of robustness or preservation to losses of these non classical features is really an important issue in case of practical implementation of information processing tasks. \\

Idea of persistency of entanglement and nonlocal features of quantum correlations under particle loss scenario (i.e., minimal number of particles to be lost for those non classical features to vanish completely) has gained interest in recent times \cite{Raussendorf2001,Brunner12,Vertesi16}. Here the loss of particles render the situation where the information about particles becomes inaccessible. For example, one can consider the multi-party quantum cryptography protocols where a number of parties are not willing to cooperate. Hence the idea of persistency of quantum correlations is crucial to the implementation of such information theoretic tasks. In\cite{Raussendorf2001}, persistency of correlation was first defined as the possibility of obtaining residual correlation when a selective measurement is performed on a subsystem. But Brunner et. al.\cite{Brunner12} have investigated the persistency of entanglement and nonlocality in a stronger scenario where a subsystem is lost completely in comparison to the earlier definition of persistency as in\cite{Raussendorf2001} for a number of special multipartite classes such as Cluster states. The authors have also discussed the possibility of maximal persistency of entanglement and nonlocality for $\mathbf{W}$ class of states. They further explored the relation between Genuine multipartite entanglement(GE) and persistence by providing an example that GE does not always imply maximal persistency. Recently Divi\'{a}nszky \emph{et. al} \cite{Vertesi16} have provided a simple upper bound on the persistency of nonlocality for $\mathbf{W}$ states and any permutation-symmetric state with two settings per party. They also provide a family of Bell inequalities which test the nonlocality under particle loss. A similar notion of persistency can also be defined for quantum steering, an weaker notion of nonlocality. This is important since the persistency of steering is required while considering star-type network for one-sided device independent quantum key distribution(DIQKD) under particle loss\cite{Branciard12}.

Keeping in mind the usefulness of genuine correlation in a multi-party scenario a pertinent question is of the persistence of GE along with GNL and similarly for genuine steering(GS)\cite{multisteer} in different classes of multipartite states. In this work we define the notion of persistency of genuine correlation and hence study the capacity of different classes of multipartite states to persist genuineness. We further investigate the possibility of achieving maximum persistency of correlations within these classes of multipartite states which is of practical interest.\\

In the following sections we first(Sec.\ref{sec1}) provide a motivation for considering the concept of persistency for genuine correlations. In Sec.\ref{sec2} we briefly provide the relevant definitions and notations. Sec.\ref{sec3} consists of our results regarding the persistency of genuine correlation for a varied class of multipartite states and the possibility of achieving maximal persistency for certain classes of multipartite states. We conclude with Sec.\ref{sec4} where we discuss the implications of our work in understanding multipartite correlations and further scopes for generalization of the results presented in this work.  

\section{Motivation}\label{sec1} 
Given that one can quantify genuine correlation of any multipartite state, then a common intuition while studying persistency of correlation under particle loss is that whichever multipartite state has the higher amount of genuine entanglement will be more persistent. But what we find out in this work is quite contrary to this intuition. We show that there exist states with less entanglement(in the sense of a valid measure of entanglement) that can have maximal persistency of entanglement whereas a class of higher entangled states has minimal persistency.

In case of multipartite nonlocality another important notion is that of monogamy. This states that all the reduced systems of a parent multi-party nonlocally correlated state obtained by tracing out every other party can not show nonlocality \cite{fei}\cite{toner}. While considering particle loss, the residual systems that can be achieved by tracing out every other party in the system are either nonlocal or not. If all of them show some nonlocality(by violation of some nonlocal inequality) then the notion of monogamy fails but the persistency of nonlocality is maintained. Thus one can consider the concept of monogamy for multi-party states as complementary to the persistency of nonlocality.

Another question that naturally arises in genuine nonlocality scenario is whether the possibility of performing local filtering operations can strictly enlarge the genuine nonlocality-persistent states at per with similar results obtained by the authors of \cite{Brunner12} in case of nonlocality and hidden nonlocality. We answer this question in affirmative and also present examples where the persistency of genuine nonlocality is 1, i.e., minimum possible; but when allowed to perform local filtering operations the persistence of 'hidden' genuine nonlocality can be maximum.

From a practical perspective, the question of achieving maximum persistency is very crucial. To understand this, consider a star network. A simple futuristic banking system is an example of a star network, where the central body bank tries to maintain quantum correlation with multiple customers. Now it is quite unexpected that since one of customers leaves the system by closing her account, the existing quantum correlation between the bank and other customers gets destroyed. So, the multipartite state shared in the star network should be something which has a higher persistency under particle loss. In this sense achieving maximum persistency is ideal. 

With these motivations in mind we move on to present our results. But before that let us discuss a few definitions and relevant tools.

\section{Definition and Tools}\label{sec2} 

\subsection{Definitions}
Consider a quantum state $\rho$ of $N$ systems. Take the partial trace over $k<N$ systems $j_1,...,j_k\in \{1,...,N\}$, and denote the reduced state 

\ba \rho_{(j_1,...,j_k)} = \text{tr}_{j_1,...,j_k} (\rho) \ea

\begin{definition} 
The strong persistency of entanglement of $\rho$, $P_E(\rho)$, is then defined as the minimal $k$ such that the reduced state $ \rho_{(j_1,...,j_k)}$ becomes fully separable, for at least one set of subsystems $j_1,...,j_k $.
\end{definition}
In \cite{Brunner12} the authors defined this stronger notion of persistency, and tried to relate them with the slightly different concept of persistency of entanglement introduced in Ref.~\cite{Raussendorf2001}. Throughout this work we have adopted this `stronger' notion of persistency, which deals with the complete loss of information of particles.

While checking for persistency of entanglement, when one considers mixed multipartite states there does not exist any necessary and sufficient criterion to detect entanglement. But in literature there are certain sufficient conditions\cite{entdet,guhne} which can be used to witness entanglement conclusively. For our purpose we make use of the criterion in \cite{guhne} to detect the presence of entanglement.

\begin{definition}  
The persistency of nonlocality of $\rho$, $P_{NL}(\rho)$, is defined in a similar way, but now demanding only that the reduced state $ \rho_{(j_1,...,j_k)}$ becomes local, i.e., that the probability distribution obtained from local measurements on $ \rho_{(j_1,...,j_k)}$ does not violate any Bell inequality. Formally this means that the probability distribution
\ba p(a_1...a_{N}|x_1...x_{N}) = \text{tr}(\rho M_{a_1}^{x_1}\otimes ... \otimes M_{a_N}^{x_N})\ea
admits a hidden variable model for general local measurement operators $M_{a_i}^{x_i}$, with $M_{a_i}^{x_i}=\openone$ if $i=j_1,...,j_k$ (the systems that have been traced out) and $\sum_{a_i} M_{a_i}^{x_i}=\openone$ otherwise. Here $x_i$ and $a_i$ denote the measurement setting and its outcome, respectively, of party $i$.
\end{definition} 
In this context one can also consider the concept of hidden nonlocality \cite{Popescu95}.  
\begin{definition} 
That is, we strengthen the above definition and demand that the reduced state $ \rho_{(j_1,...,j_k)}$ is local even after the remaining parties have performed a local filtering. In this case, persistency of nonlocality is denoted by $P_{HNL}(\rho)$.
\end{definition} 

For any state $\rho$, one has that
\ba N-1\geq P_E(\rho) \geq P_{HNL}(\rho) \geq P_{NL}(\rho) \geq 1. \ea
The second inequality comes from the fact that (i) entanglement is necessary for having quantum nonlocality, and (ii) there exist entangled states which are local \cite{Werner89}. The third inequality follows from the fact that there exist local quantum states featuring hidden nonlocality \cite{Popescu95}.

In a similar spirit one can also define the concept of persistency of steering as follows:
\begin{definition}\label{persteer}
The persistency of steering of $\rho$, $P_S(\rho)$, is defined as the minimal $k$ such that the reduced state $ \rho_{(j_1,...,j_k)}$ becomes fully unsteerable, for at least one set of subsystems $j_1,...,j_k $.
\end{definition}
Then the revised hierarchy of the persistency of a state $\rho$ will be given by
\ba N-1\geq P_E(\rho) \geq P_{S}(\rho) \geq P_{NL}(\rho) \geq 1. \ea
The second inequality follows from the fact that entanglement is necessary for having quantum steering and there exist entangled states which are unsteerable\cite{Quintino15}. The third inequality follows since quantum steering is necessary for having nonlocality and there exist steerable states which are local\cite{Wiseman07}\cite{Quintino15}.

The notions described above have a large impact from an operational angle. Here we try to characterize the robustness of multipartite quantum correlations under loss of particles. In particular, experimentally how entanglement behaves under particle loss has been investigated  but their experiment deals with multi-qubit states \cite{mohamed2006}. Clearly, if violation of a Bell inequality is observed for all the reduced states of $\rho$  under the condition where $k$ parties are traced out, then it is guaranteed that all reduced states possesses entanglement, irrespectively of the Hilbert space dimension of the state and the alignment of the measurement devices.

In case of more than two parties one has different notions of correlation. The concept of genuine correlation provides an interesting paradigm to understand correlation in multipartite scenario by excluding the possibility of bipartite correlations. At this point let us provide the definitions for genuine entanglement, nonlocality and steering. For simplicity we define these notions for three parties but reader can easily extend these definitions for higher number of parties. 

\begin{definition}
A quantum state is bi-separable if it can be written as
$\rho_{ABC}=\sum_{\lambda}p_{A(BC)}^\lambda\rho_{A}^\lambda\otimes \rho_{BC}^\lambda+\sum_{\mu}p_{B(AC)}^\mu\rho_{B}^\mu\otimes \rho_{AC}^\mu+\sum_{\nu}p_{C(AB)}^\nu\rho_{C}^\nu\otimes \rho_{AB}^\nu$ where $p_{A(BC)}^\lambda$, $p_{B(AC)}^\mu$ and $p_{C(AB)}^\nu$ are probability distributions. Finally a state is genuine multipartite entangled, if it is not bi-separable.
\end{definition}
There does not exist any necessary and sufficient criterion to detect genuine entanglement of mixed multipartite states. But several sufficient conditions for witnessing genuine entanglement have been proposed\cite{entdet,guhne}. Here we make use of the sufficient criteria given in \cite{guhne} for our purpose.

Now let us consider the correlation scenario among three parties with inputs and outputs as $\{X,Y,Z\}$ and $\{a,b,c\}$ respectively. Then one can have the following definition for genuine nonlocal correlations\cite{Bancal}-
\begin{definition}
Suppose that $P(abc|XYZ)$ can be written in the form 
\begin{align}\label{genuinenonsignonloc}
&P(abc|XYZ) = \sum_\lambda q_\lambda P_{\lambda}(ab|XY)P_{\lambda}(c|Z)+\nonumber \\ 
&\sum_\mu q_\mu P_{\mu}(ac|XZ)P_{\mu}(b|Y)+\sum_\nu q_\nu P_{\nu}(bc|YZ)P_{\nu}(a|X),
\end{align}
where the bipartite terms are non-signalling. Then the correlations are $NS_2$-local. Otherwise, we say that they are genuinely 3-way nonlocal.
\end{definition}

In a recent work\cite{Jeba16} the concept of genuine steering among three parties have been defined as following-
\begin{definition}
Suppose that $P(abc|XYZ)$ cannot be explained by the following nonlocal LHS-LHV (NLHS)
model:
\begin{align}
P(abc|XYZ)&=\sum_\lambda p_\lambda P(ab|XY,\rho_{AB}^\lambda)P_\lambda(c|Z)+\nonumber \\
&\sum_\lambda q_\lambda P(a|X,\rho_A^\lambda)P(b|Y,\rho_B^\lambda)P_\lambda(c|Z), \label{NLHS}
\end{align}
where $P(ab|XY,\rho_\lambda)$ denotes the nonlocal probability distribution arising from two-qubit state $\rho_{AB}^\lambda$, and $P(a|X,\rho_A^\lambda)$ 
and $P(b|Y,\rho_B^\lambda)$ are the distributions arising from qubit states $\rho_A^\lambda$ and $\rho_B^\lambda$ and $\{p_\lambda\}$, $\{q_\lambda\}$ are probability distributions. 
Then the quantum correlation exhibits genuine steering from Charlie to Alice and Bob.
\end{definition}

Based on the definitions of genuine correlation presented above, one can also define the following quantities regarding the persistency of genuineness in correlations: 
\begin{definition}\label{perge}
Persistency of genuine entanglement ($P_{GE}$), nonlocality ($P_{GNL}$) and steering ($P_{GS}$) for a quantum state is defined as the minimum number of particle lost so that the reduced state is no longer genuinely entangled, nonlocal and steerable respectively.  
\end{definition}
\begin{definition}\label{perhgnl}
Persistency of genuine nonlocality under local filtering operation ($P_{HGNL}$) for a quantum state is defined as the minimum number of particle lost so that the reduced state is no longer genuinely nonlocal under local filtering operations.  
\end{definition}

One can clearly see that $P_{GE}\ge P_{GS}\ge P_{GNL}$. 
The first inequality follows from the fact genuine entanglement is necessary for genuine steering. The second inequality comes from the requirement that genuine steering is necessary for genuine nonlocality and there exist genuine steerable states which are not genuinely nonlocal\cite{Jeba16}.

\section{Results}\label{sec3} 
In these section we present our results regarding persistency of genuine correlations for a number of classes of multi-party states. These states are important for different information theoretic tasks. Hence the robustness of genuine correlation under particle loss for these states are of practical importance.  
\subsection{A generic class of four qubit states}
It was argued in~\cite{Ver02} that 4-qubits pure states can be classified into nine groups of states. One of these nine groups is called the generic class as with the action of stochastic local operations with classical communication(SLOCC) it is dense in the space of four qubits 
$\mathcal{H}_4\equiv \mathbb{C}^2\otimes\mathbb{C}^2\otimes\mathbb{C}^2\otimes\mathbb{C}^2$.
The generic class is given by
$$
\mathcal{A}\equiv\Big\{z_0u_0+z_1u_1+z_2u_2+z_3u_3\Big|\;z_0,z_1,z_2,z_3\in\mathbb{C}\Big\}\;.
$$ 
where 
\begin{align}
& u_0\equiv|\phi^{+}\rangle|\phi^{+}\rangle\;\;,\;\;u_1\equiv|\phi^{-}\rangle|\phi^{-}\rangle\nonumber\\
& u_2\equiv|\psi^{+}\rangle|\psi^{+}\rangle\;,\;\;
u_3\equiv|\psi^{-}\rangle|\psi^{-}\rangle\nonumber
\end{align}
A pure state of 4 qubits $|\psi\rangle \in \mathcal{A}$ can be written in computational basis as the following: {\small
\begin{eqnarray*}
|\psi\rangle &=&\frac{z_0+z_3}{2}(|0000\rangle
+|1111\rangle)+\frac{z_0-z_3}{2}(|0011\rangle
+|1100\rangle)\\
&&\hspace{.1cm}+\frac{z_1+z_2}{2}(|0101\rangle
+|1010\rangle)+\frac{z_1-z_2}{2}(|0110\rangle +|1001\rangle)
\end{eqnarray*}

The 4-qubits entanglement monotone that is invariant under any permutation of the 4-qubits, the Wong-Christensen 4-tangle~\cite{Uh00} is defined as the following. Let 
$|\psi\rangle\in\mathcal{H}_4\equiv\mathbb{C}^2\otimes\mathbb{C}^2\otimes\mathbb{C}^2\otimes\mathbb{C}^2$, the 4-tangle is 
defined by~\cite{Uh00}
\begin{equation}
\tau_{ABCD}\equiv |\langle\psi|\sigma_y\otimes\sigma_y\otimes\sigma_y\otimes\sigma_y|\psi^{*}\rangle|^2.
\end{equation} 
As a measure for pure bipartite entanglement we first take the tangle or the square of the I-concurrence \cite{deftangle}. In four qubits there are 4 bipartite cuts consisting of one-qubit verses the rest three quibts and three bi-partite cuts consisting of 2-qubits in each cut. Denoting the four qubits by A, B, C, and D, one can define
\begin{align}
\tau_1 & \equiv\frac{1}{4}\left(\tau_{A(BCD)}+\tau_{B(ACD)}+\tau_{C(ABD)}+\tau_{D(ABC)}\right)\label{t1}\\
\tau_2 & \equiv\frac{1}{3}\left(\tau_{(AB)(CD)}+\tau_{(AC)(BD)}+\tau_{(AD)(BC)}\right)\;,\label{t2}
\end{align}
where $\tau_{A(BCD)}$, is the tangle between qubit A and qubits B,C,D. Similarly, $\tau_{(AB)(CD)}$, is the tangle between qubits A,B and qubits C,D.  Reader can note that the maximum value possible for $\tau_1$ is $1$ and the maximum value possible for $\tau_2$ is $3/2$. This is since a maximally entangled $4\times 4$ bipartite state has tangle $3/2$. 
However, no 4 qubit pure state can achieve this value for $\tau_2$. In \cite{GBS}, it has been shown that 
\begin{equation}
\tau_1\leq\tau_2\leq\frac{4}{3}\tau_1\;.
\label{bounds}
\end{equation}
Hence, it follows that $\tau_2\leq 4/3<3/2$(since $\tau_1$ is bounded by $1$). Moreover, from the inequality above, it follows that for states with $\tau_2=4/3$, $\tau_1$ must be equal to $1$.

In \cite{GBS} it was shown that if $|\psi\rangle\in\mathcal{H}_4\equiv\mathbb{C}^{2}\otimes\mathbb{C}^{2}\otimes\mathbb{C}^{2}\otimes\mathbb{C}^{2}$ be a normalized four qubit state, then
$$
\tau_1\left(|\psi\rangle\right)=1\;\;\;\text{if and only if}\;\;\;|\psi\rangle\in\mathcal{A}\;,
$$
up to local unitary transformation.

In~\cite{Ver02} the authors argued that among all the 4-qubit pure states, only states in $\mathcal{A}$ have $\tau_1=1$. Later in \cite{Ver03} this was fully proved.

Thus it follows that among all the states with $\tau_1=1$, we have
$$ 1\leq\tau_2\leq\frac{4}{3}\;.$$
It is interesting to note that the 4-qubit GHZ state gives the minimum possible value for $\tau_2$, i.e. it is the least entangled state among all the states with $\tau_1=1$.

For all 4-qubits pure states $\tau_{ABCD}=4\tau_1-3\tau_2$. $4\tau_1$ can be interpreted as the total amount of entanglement in the multi-party system, whereas $3\tau_2$ can be interpreted as the total amount of entanglement shared among groups consisting of two qubits each and thus the 4-tangle can be interpreted as the residual entanglement that can not be shared among the two qubits groups\cite{GBS}.

\subsubsection{$\tau_{min}$ and $\mathcal{M}$ classes}

A normalized state $|\psi\rangle\in\mathcal{H}_4$ is maximally entangled (i.e. $\tau_2(|\psi\rangle)=4/3$)
if and only if up to local unitary $|\psi\rangle\in\mathcal{M}$, where $\mathcal{M}$ is the set of states in $\mathcal{A}$ 
with zero 4-tangle\cite{GBS}.

As defined earlier a pure state $\psi=\sum_{j=0}^{3}z_ju_j$ in $\mathcal{A}$ depends on four complex parameters $z_j$ $(j=0,1,2,3)$. The condition that the 4-tangle 
$\tau_{ABCD}(\psi)=|\sum_{j=0}^{3}z_{j}^{2}|^2=0$ implies that 
the states in the maximally entangled class $\mathcal{M}$ are characterized by 4 \emph{real} parameters. The reduction in number of parameters is due to the normalization condition and ignoring the global phase. When written in its polar form $z_j=\sqrt{p_j}e^{i\theta_j}$ (with non-negative $p_j$ and $\theta_j\in[0,2\pi]$)
one can denote the class $\mathcal{M}$ as follows:
\begin{equation}
\mathcal{M}=\left\{\sum_{j=0}^{3}\sqrt{p_j}e^{i\theta_j}u_j\Big|\;\sum_{j=0}^{3}p_j=1\;,\; \sum_{j=0}^{3}p_je^{2i\theta_j}=0\right\}\;.
\label{classm}
\end{equation} 
 
Another important set of the states in $\mathcal{A}$, denoted by $\mathcal{T}_{\min}$, with the minimum possible value $\tau_2=1$, can be characterized as follows:
\begin{align}
\mathcal{T}_{\min} & \equiv\left\{\psi\in\mathcal{A}\Big|\tau_2(\psi)=1\right\}\nonumber\\
& =\left\{\sum_{j=0}^{3}x_ju_j\Big|\;\sum_{j=0}^{3}x_{j}^{2}=1\;,\;x_j\in\mathbb{R}\right\}
\end{align}
For example, the four qubits GHZ state belongs to $\mathcal{T}_{\min}$. In this sense, the GHZ state is a state in $\mathcal{A}$ with the least amount of entanglement.
\subsubsection{Persistency of Entanglement and Genuine Entanglement}
{\bf $P_E$ and $P_{GE}$ of $\tau_{min}$ class-} One can provide the following conditions for $P_{GE}$ and $P_{E}$ in the $\tau_{min}$ class-
\begin{condition}\label{taumin1} 
$P_{GE}>1$ if 
\begin{eqnarray*}
2|x_2^2-x_3^2|> && 2[x_2^2+x_3^2]+(x_0+x_1)^2+2[x_0^2-x_1^2]\\&&-8 \min \{|x_0|,|x_1|\}\max\{|x_2|,|x_3|\}
\end{eqnarray*}
\end{condition} 

\begin{condition}\label{taumin2} 
$P_{E}>1$ if
\begin{itemize}
\item for $sgn(x_0x_1)=1$ 
\begin{equation*}
|x_2^2-x_3^2|> [x_0^2-x_1^2]+4 \min\{|x_0|,|x_1|\}\max\{|x_2|,|x_3|\}
\end{equation*}
\item for $sgn(x_0x_1)=-1$ 
\begin{equation*}
|x_2^2-x_3^2|> [x_0^2-x_1^2]-4 \min\{|x_0|,|x_1|\}\max\{|x_2|,|x_3|\}
\end{equation*}
\end{itemize}
\end{condition}

At this point one might wonder whether the $\tau_{min}$ class contains states with minimal persistency of entanglement. Let us consider the following example of the state $|GHZ_4\rangle=\frac{1}{\sqrt{2}}[|0000\rangle+|1111\rangle]$. It is straightforward to show(see in Appendix) that both $P_E$ and $P_{GE}$ of $|GHZ_4\rangle$ are $1$. This implies that even for the same value of the entanglement measure $\tau_2$ throughout the $\tau_{min}$ class there exist states which have different capabilities of persisting entanglement and genuine entanglement.   

{\bf $P_E$ and $P_{GE}$ of $\mathcal{M}$ class-} Intuitively it can be expected that $\mathcal{M}$ class states being maximally entangled might have greater persistency of entanglement compared to the states in the $\tau_{min}$ class. Let us take the example of cluster states. Cluster states\cite{Raussendorf2001} form a class of multi-party entangled quantum states with surprising and useful properties. The main interest in these states draws from their role as a universal resource in the one-way quantum computer\cite{comp}: Given a collection of sufficiently many particles that are prepared in a cluster state, one can realize any quantum computation by simply measuring the particles, one by one, in a specific order and basis. By the measurements, one exploits correlations in quantum mechanics which are rich enough to allow for universal logical processing. A four party cluster state is given by the following
\begin{equation}
\eta_4=\frac{1}{2}\left[|0000\rangle+|0011\rangle+|1100\rangle-|1111\rangle\right]  \nonumber  
\end{equation}
 
The tripartite reduced states can be written in a bi-separable form\cite{Brunner12}. Thus $P_{GE}=1$ for cluster states. This is in contrast to the states in the $\tau_{min}$ class which are less entangled according to the measure $\tau_2$ but can have $P_{GE}>1$ according to Condition.\ref{taumin1}-\ref{taumin2}.    

\subsubsection{Persistency of Nonlocality and Genuine Nonlocality}
Now let us come to the question of persistence of nonlocality for states in the $\tau_{min}$ class.   
\begin{theorem}
$P_{NL}(\rho)=1$ for all four qubit states $\rho\in \tau_{min}$. 
\end{theorem}

\begin{proof}
Let us consider a four qubit state $\rho\in \tau_{min}$. Upon loss of $i$-th particle the reduced states are:
\begin{equation*}
\rho_{i}^3=|\psi_{i}^3\rangle\langle\psi_{i}^3|+|\phi_{i}^3\rangle\langle\phi_{i}^3|
\end{equation*}
for $i=1,2,3,4$ where, $|\psi_{i}^3\rangle=|\tilde{W}_{i}^3\rangle+\frac{x_0+x_1}{2}|111\rangle $, $|\phi_{i}^3\rangle=\sigma_x^{\otimes 3}\left[ |\tilde{W}_{i}^3\rangle+\frac{x_0+x_1}{2}|111\rangle\right]  $ and 
\begin{eqnarray*}
|\tilde{W}_{1}^3\rangle&=&\frac{x_2-x_3}{2}|001\rangle+\frac{x_2+x_3}{2}|010\rangle+\frac{x_0-x_1}{2}|100\rangle\\
|\tilde{W}_{2}^3\rangle&=&\frac{x_2+x_3}{2}|001\rangle+\frac{x_2-x_3}{2}|010\rangle+\frac{x_0-x_1}{2}|100\rangle\\
|\tilde{W}_{3}^3\rangle&=&\frac{x_0-x_1}{2}|001\rangle+\frac{x_2-x_3}{2}|010\rangle+\frac{x_2+x_3}{2}|100\rangle\\
|\tilde{W}_{4}^3\rangle&=&\frac{x_0-x_1}{2}|001\rangle+\frac{x_2+x_3}{2}|010\rangle+\frac{x_2-x_3}{2}|100\rangle
\end{eqnarray*}
To check the nonlocality of these reduced tripartite states let us consider all $46$ facets of the three qubit local polytope\cite{Sliwa03}. One can check that all the reduced states $\rho_{i}^3$ violate only the $4$-th facet (same numbering as in \cite{Sliwa03} has been used for convenience) for different values of the real parameters $\{x_i\}_{i=0}^3$. At the same time it can also be shown (see Appendix.(\ref{4thfacet})) that all reduced states can not violate the $4$-th facet for a common set of parameter values. This implies that the nonlocality of any $\rho\in \tau_{min}$ can not be persisted upon loss of even one of the particles. Hence the theorem.    
\end{proof}
From this theorem one can immediately arrive at the following corollary regarding the persistency of genuine nonlocality
\begin{corollary}
$P_{GNL}(\rho)=1$ for all four qubit states $\rho\in \tau_{min}$. 
\end{corollary}

At this stage a pertinent question would be whether a weaker form of nonlocality can persist upon loss of particles for states in the $\tau_{min}$ class. We deal with this question in the next subsection considering quantum steering as a weaker form of nonlocality.        

\subsubsection{Persistency of Steering and Genuine Steering}
\begin{observation}
There exist states $\rho\in \tau_{min}$ such that $P_{S}(\rho)$ is maximal, i.e. 3. 
\end{observation}
This can be seen in a straightforward way. If one can show that there exist two qubit reduced states which can demonstrate steering, this in turn implies that there exist states in the $\tau_{min}$ class with maximal persistency of steering. Upon loss of two particles the bipartite reduced states take the following forms:

\begin{equation*}
\rho_{i}^2 = |\eta_{i}^2\rangle\langle \eta_{i}^2|+ |\xi_{i}^2\rangle\langle \xi_{i}^2|+\sigma_x^{\otimes 2}|\eta_{i}^2\rangle\langle \eta_{i}^2|\sigma_x^{\otimes 2}+\sigma_x^{\otimes 2}|\xi_{i}^2\rangle\langle \xi_{i}^2|\sigma_x^{\otimes 2}   
\end{equation*}   
for $i=1,2,3$, where,  
\begin{eqnarray}
|\eta_{1}^2\rangle &=& \frac{x_0-x_1}{2}|00\rangle+ \frac{x_0+x_1}{2}|11\rangle\nonumber \\
|\eta_{2}^2\rangle &=& \frac{x_2+x_3}{2}|00\rangle+ \frac{x_0+x_1}{2}|11\rangle\nonumber \\
|\eta_{3}^2\rangle &=& \frac{x_2-x_3}{2}|00\rangle+ \frac{x_0+x_1}{2}|11\rangle\nonumber 
\end{eqnarray}
and 
\begin{eqnarray}
|\xi_{1}^2\rangle &=& \frac{x_2-x_3}{2}|01\rangle+ \frac{x_2+x_3}{2}|10\rangle\nonumber \\
|\xi_{2}^2\rangle &=& \frac{x_2+x_3}{2}|01\rangle+ \frac{x_0-x_1}{2}|10\rangle\nonumber \\
|\xi_{3}^2\rangle &=& \frac{x_2+x_3}{2}|01\rangle+ \frac{x_0-x_1}{2}|10\rangle\nonumber 
\end{eqnarray}

Now there exist states $\rho^4\equiv\{x_0,x_1,x_2,x_3\}$ such that $\rho_{i}^2$ is steerable for $i=1,2,3$ (see Appendix.(\ref{existsteer})). The existence of such states can be depicted in the parameter space as shown in Fig.(\ref{maxsteer}).
\begin{figure}[h!]\label{maxsteer} 
\includegraphics[scale=0.35]{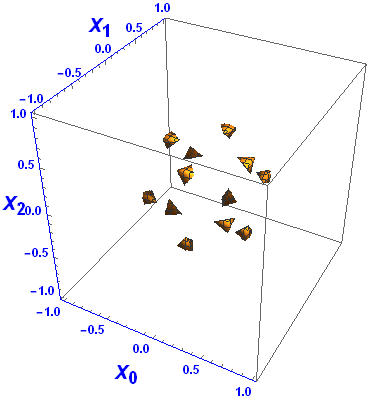}
\caption{This figure shows the states with $P_S=3$ in the parameter space $\{x_0,x_1,x_2\}$.}
\end{figure}  

\subsection{Achieving maximal persistency}
Let us now consider the cases of achieving maximal persistency of correlations. These are interesting since it signifies the robustness of the correlation under particle loss and the failure of monogamy of correlation between distant parties.  
\subsubsection{Maximal persistency of Genuine Nonlocality}
In \cite{Brunner12} the authors could not present any state with local dimension $2$ which has maximum persistency of nonlocality. This can partly be understood as the strength of monogamy principle for nonlocality\cite{fei}\cite{toner}. This implies that the demonstration of maximal persistency of genuine nonlocality will be harder. But there exist multipartite states with local dimension $2$ which can demonstrate maximal persistency of genuine nonlocality when local filtering is allowed. This is to say $P_{HGNL}$ for such states are maximal. Let us consider the following example:
\begin{example}
An $n$-partite state $|W^N\rangle$ is given by 
\begin{equation*}
|W^N\rangle=\frac{1}{\sqrt{N}}\left[ |0...01\rangle+|0...10\rangle+...+|10...0\rangle\right] 
\end{equation*} 
Now consider any reduced state of three parties obtained by loosing $(N-3)$ parties. These states are of the form:
\begin{equation}\label{rho3}
\rho(p)=p|W^3\rangle\langle W^3|+(1-p)|000\rangle\langle 000|
\end{equation}
where $p=\frac{3}{N}$. These reduced states do not demonstrate genuine nonlocality for two settings per site since they do not violate any of the 185 inequalities given in \cite{Bancal}. Thus persistency of genuine nonlocality for $W^N$ cannot be maximum. Now take the local filtering of the form
\[
\begin{bmatrix}
\epsilon & 0\\
0 & 1
\end{bmatrix}
\]
where $0\le \epsilon\le 1$. After local filtering the state becomes
\begin{equation}\label{rhop}
\rho(p,\epsilon)=\frac{p\epsilon^4}{p\epsilon^4+(1-p)\epsilon^6}|W^3\rangle\langle W^3|+\frac{(1-p)\epsilon^6}{p\epsilon^4+ (1-p)\epsilon^6}|000\rangle\langle 000|
\end{equation} 
Now consider the Bell quantity
\begin{eqnarray}
B_{16}=&&\langle A_0B_0\rangle+\langle A_1B_0\rangle+\langle A_0B_1\rangle-\langle A_1B_1\rangle-2\langle C_0\rangle\nonumber\\&&+\langle A_0B_0C_0\rangle+\langle A_1B_0C_0\rangle+\langle A_1B_1C_0\rangle\nonumber\\&&+2\langle A_1C_1\rangle+2\langle B_1C_1\rangle
\end{eqnarray}
where $B_{16}\le 4$ is the $16$-th facet inequality as given in \cite{Bancal}. The maximum value of $B_{16}$ obtainable from a state of the form (\ref{rhop}) is 
\begin{equation}
B_{16}(p,\epsilon)=\frac{p(4.72678)+2\epsilon^2 (p-1)}{\epsilon^2 (1-p)+p}
\end{equation}
By choosing $\epsilon\rightarrow 0$, this value can reach upto $4.72678$ for any value of $p$. This implies $P_{HGNL}(W^N)\ge (N-2)$. In \cite{Brunner12} it has already been shown that the bipartite reduced states of $|W^N\rangle$ exhibits hidden nonlocality. Thus one has $P_{HGNL}(W^N)=(N-1)$, i.e. maximal persistency of genuine correlation under local filtering.   
\end{example} 
But it would be more interesting to find the persistency of genuine nonlocality when local filtering is not considered. Since the three qubit reduced state $\rho(p)$ does not violate any of the 185 facets for detecting genuine nonlocality under two measurement settings for all parties. On the basis of this evidence we make the following conjecture      
\begin{conjecture}
$P_{GNL}(W^N)<N-2$ for all N-partite {\bf W} states with local dimension $2$. 
\end{conjecture}

In the next subsection we ask the question whether a weaker form of genuine nonlocality namely, genuine quantum steering can achieve maximal persistency.     
\subsubsection{Maximal persistency of Genuine Steering}

Here we present a 4 qubit state which exhibits maximum
persistency of genuine steering. We present our argument below.
Let us consider the state $|W^4\rangle$. Remember that this state does not achieve maximal persistency of nonlocality or genuine nonlocality. Upon loss of one particle the three qubit reduced state take the form (\ref{rho3}), where $p=\frac{3}{4}$. This state violates the following 3-setting genuine steering inequality(see Appendix(\ref{gensteer})):
\begin{equation}
|\langle D_0C_0\rangle+\langle D_1C_1\rangle+\langle D_2C_2\rangle|\le 3
\end{equation}   
where,
\begin{eqnarray}
D_0&=& A_0B_0+A_1B_1+A_2B_2\\
D_1&=& A_0B_2-A_1B_0+A_2B_1\\
D_2&=& A_0B_1-A_1B_2+A_2B_0
\end{eqnarray}
Thus one has $P_{GS}(W^4)\ge 2$. Now for loss of two parties the two qubit reduced state of $|W^4\rangle$ is of the form:
\begin{equation}
\rho^2(p)=p|W^2\rangle\langle W^2|+(1-p)|00\rangle\langle 00|
\end{equation}
where, $p=\frac{1}{2}$. We know that this state has a local model under projective measurements\cite{Tran14}. But nonetheless this state exhibits steering because it violates a sufficient criterion\cite{Jevtic15} for steering. Hence $|W^4\rangle$ has maximum persistency of genuine steering, i.e. $P_{GS}(W^4)=3$. This trivially implies the persistency of steering of $|W^4\rangle$ state is also maximum.

\subsubsection{Maximal persistency of Genuine Entanglement}
Now let us come to the question of maximum persistency of genuine entanglement. Let us consider the 4 qubit $\tau_{min}$ class of states. States belonging to this class will have maximum persistency of genuine entanglement i.e. $P_{GE}=3$ under the following conditions:
\begin{condition}
$P_{GE}>1$ and $S_i>0$ for $i=1,2,3$
\end{condition}
where,
\begin{eqnarray}
S_1&=& 2\max \{|-\frac{1}{2}(x_0+x_1)(-x_2+x_3)|-\frac{1}{4}((x_0-x_1)^2+(x_2+x_3)^2),\nonumber\\
&&|\frac{1}{2}(x_0-x_1)(x_2+x_3)|-\frac{1}{4}((x_0+x_1)^2+(x_2-x_3)^2)\}\nonumber \\ 
S_2&=& 2\max \{|\frac{1}{2}(x_0+x_1)(x_2+x_3)|-\frac{1}{4}((x_0-x_1)^2+(x_2-x_3)^2),\nonumber\\
&&|-\frac{1}{2}(x_0-x_1)(-x_2+x_3)|-\frac{1}{4}((x_0+x_1)^2+(x_2+x_3)^2)\}\nonumber \\
S_3&=& 2\max \{|\frac{1}{2}(x_2+x_3)(x_2-x_3)|-\frac{1}{4}((x_0-x_1)^2+(x_0+x_1)^2),\nonumber\\
&&|\frac{1}{2}(x_0+x_1)(x_0-x_1)|-\frac{1}{4}((x_2+x_3)^2+(x_2-x_3)^2)\}\nonumber  
\end{eqnarray}
We depict the states with maximum persistency of genuine entanglement in the $\tau_{min}$ class in the parameter space in Fig.(\ref{maxent}).
\begin{figure}[h!]\label{maxent} 
\includegraphics[scale=0.35]{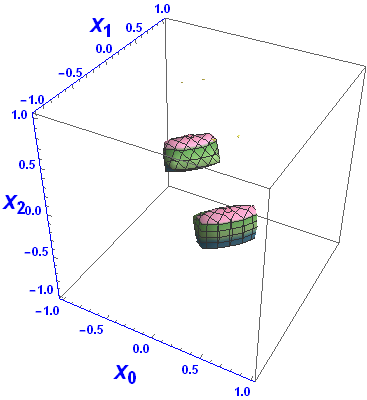}
\caption{This figure shows the states with $P_{GE}=3$ in the parameter space $\{x_0,x_1,x_2\}$.}
\end{figure}  
For example, the 4-qubit Dicke state $|D^4\rangle$ \cite{Dicke} belongs to the $\tau_{min}$ class and it satisfies all the above conditions and hence $P_{GE}(D^4)=3$.

The conditions for states in $\tau_{min}$ class to have maximum persistency of entanglement are the following:
\begin{condition}
$S_i>0$ for $i=1,2,3$
\end{condition}
The states with $P_E=3$ belonging to the $\tau_{min}$ class has been shown in Fig.(\ref{maxe}) 
\begin{figure}[h!]\label{maxe} 
\includegraphics[scale=0.35]{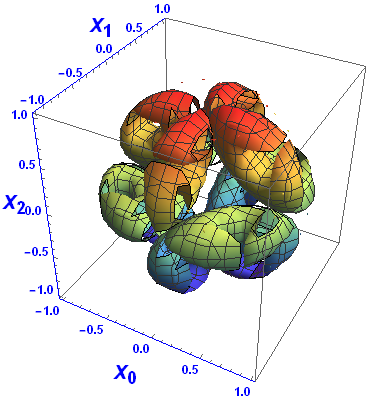}
\caption{This figure shows the states with $P_{E}=3$ in the parameter space $\{x_0,x_1,x_2\}$.}
\end{figure}

\section{Discussions}\label{sec4} 
In a couple of recent studies\cite{Raussendorf2001}\cite{Brunner12} the concept of persistency of entanglement and nonlocality were introduced. This new concept is fundamental to the understanding of quantum correlations and at the same time important from practical perspective since it deals with the scenario where information about some of the parties can be completely lost. 

Besides defining the same notion for quantum steering, our work extends the concept of persistency to genuine correlations which are inherently multipartite in nature. We also discuss the possibility of achieving maximum persistency of genuine correlations with several important classes of multipartite states. As we have emphasized in the subsequent sections that maximum persistency of correlation becomes indispensable in certain multi-party quantum cryptography protocols. Now we point out some of the questions which this works leaves open. A thorough understanding of the persistency of correlation for the four qubit states can enable one to classify the whole class of four qubit states in terms of persistency. Moreover, one could also extend the study of persistency of correlation for multi-party systems of higher dimension($>2$). Another interesting question is to find out multi-party qubit states which have maximal persistency of genuine nonlocality.           
	
\emph{Acknowledgement}:
We would like to thank Prof. Guruprasad Kar for useful discussions. AM acknowledge support from the CSIR project 09/093(0148)/2012-EMR-I.
\appendix
\section{Violation of 4-th facet}\label{4thfacet} 
Maximum violation value of each reduced state for $4$-th facet\cite{Paul14} :
\begin{eqnarray}\label{A1}
B_{Max}(\rho_1^{3}) = && \max[2 \sqrt{{((-1 + 2 x_1^2))^2 + ((-1 + 2 x_0^2) )^2}},\nonumber\\
 && 2 \sqrt{{((-1 + 2 x_0^2 + 2 x_1^2))^2 + ((1 - 2 x_0^2) )^2}},\nonumber\\
&& 2 \sqrt{{((-1 + 2 x_0^2 + 2 x_1^2))^2 + ((-1 + 2 x_1^2) )^2}}]
 \end{eqnarray} 
\begin{eqnarray}\label{A2} 
 B_{Max}(\rho_2^{3}) = && \max[2 \sqrt{{((1 - 2 x_0^2  - 2 x_2^2 ))^2 + ((1 - 2 x_1^2  - 2 x_2^2 ))^2}}\nonumber\\,
  && 2 \sqrt{{((1 - 2 x_1^2 - 2 x_0^2 ))^2 + ((1 - 2 x_2^2 - 2 x_0^2 ))^2}},\nonumber\\
 && 2 \sqrt{{((1 - 2 x_0^2  - 2 x_1^2 ))^2 + ((1 - 2 x_2^2  - 2 x_1^2 ))^2}}]
\end{eqnarray}

\begin{eqnarray}\label{A3} 
B_{Max}(\rho_3^{3}) = && \max[4 \sqrt{{(x_0 x_2 + x_1 x_3)^2 + (x_1 x_2 + x_0 x_3)^2}},\nonumber\\ 
&& 4 \sqrt{{(x_0 x_1 + x_2 x_3)^2 + (x_1 x_2 + x_0 x_3)^2}},\nonumber\\ 
&& 4 \sqrt{{(x_0 x_2 + x_1 x_3)^2 + (x_1 x_0 + x_2 x_3)^2}}]
\end{eqnarray}

\begin{eqnarray}\label{A4}      
B_{Max}(\rho_3^{4}) = && \max[4 \sqrt{{(x_0 x_1 - x_2 x_3)^2 + (x_0 x_2 - x_1 x_3)^2}},\nonumber\\
 && 4 \sqrt{((x_0 x_2 - x_1 x_3))^2 + (x_1 x_2 - x_0 x_3)^2},\nonumber\\
 && 4 \sqrt{((x_1 x_2 - x_0 x_3))^2 + (x_1 x_0 - x_2 x_3)^2}]
\end{eqnarray}
     
 All these reduced states exhibit nonlocality only when each of $\rho_i^{3} > 2 (i= 1,2,3,4)$. It is impossible that all reduced states violate $4$-th facet inequality. Hence the result $P_{NL} = 1$.

\section{Existence of maximally persistent steerable states}\label{existsteer}
The following conditions on the state parameters are obtained from the conditions given in \cite{Jevtic15}
 \begin{itemize}
   \item Steering condition of $\rho_1^{2}$: 
   \begin{eqnarray} 
   S(\rho_1^{2}):= && \max[|-1 + 2 x_0^2 + 2 x_2^2| + |-1 + 2 x_1^2 + 2 x_2^2|\nonumber\\ &&-\frac{2}{\pi} (\sqrt{1 - (-1 + 2 x_0^2 + 2x_1^2)^2}+\sqrt{1 - (-1 + 2 x_0^2 + 2 x_1^2)^2}),\nonumber\\&&|-1 + 2 x_0^2 + 2 x_2^2| + |-1 + 2 x_0^2 + 2 x_1^2|\nonumber\\ &&-\frac{2}{\pi} (2 \sqrt{(1)^2 - (-1 + 2 x_1^2 + 2 x_2^2)^2}),\nonumber\\&&|-1 + 2 x_1^2 + 2 x_2^2| + |-1 + 2 x_0^2 + 2 x_1^2|\nonumber\\ &&-\frac{2}{\pi} (2 \sqrt{1 - (-1 + 2 x_0^2 + 2 x_2^2)^2})]
   \end{eqnarray} 
   \item Steering condition of $\rho_2^{2}$:
   \begin{eqnarray}
   S(\rho_2^{2}):=&&\max[|2 (x_0 x_2 + x_1 x_3)| +|-2 (x_1 x_2 + x_0 x_3)\nonumber\\ &&-\frac{2}{\pi} (\sqrt{1 - (2 (x_0 x_1 + x_2 x_3))^2} + \sqrt{1 - (2 (x_0 x_1 +x_2 x_3))^2}),\nonumber\\&&|2 (x_0 x_2 - x_1 x_3)| + |2 (x_0 x_1 + x_2 x_3)|\nonumber\\ &&-\frac{2}{\pi}  (2 \sqrt{(1)^2 - (-2 (x_1 x_2 + x_0 x_3))^2}),\nonumber\\&&|-2 (x_1 x_2 + x_0 x_3)| +|2 (x_0 x_1 + x_2 x_3)|\nonumber\\ &&-\frac{2}{\pi}  (2 \sqrt{1 - (2 (x_0 x_2 + x_1 x_3))^2})]
   \end{eqnarray}          
   \item Steering condition of $\rho_3^{2}$: 
   \begin{eqnarray}
   S(\rho_3^{2}):=&&\max[|2 x_0 x_2 - 2 x_1 x_3| +|-2 x_1 x_2 + 2 x_0 x_3|\nonumber\\ &&-\frac{2}{\pi} (\sqrt{1 - (2 x_0 x_1 -2 x_2 x_3)^2} +
\sqrt{1 - (2 x_0 x_1 -2 x_2 x_3)^2}),\nonumber\\&&|2 x_0 x_2 - 2 x_1 x_3| +|2 x_0 x_1 - 2 x_2 x_3|\nonumber\\ &&-\frac{2}{\pi} (2 \sqrt{1 - (-2 x_1 x_2 + 2 x_0 x_3)^2}),\nonumber\\&&|-2 x_1 x_2 + 2 x_0 x_3| +|2 x_0 x_1 - 2 x_2 x_3|\nonumber\\ &&-\frac{2}{\pi} (2 \sqrt{1 - (2 x_0 x_2 -2 x_1x_3)^2})]
 \end{eqnarray} 
 \end{itemize}
All of these three reduced states exhibit steering if $S(\rho_1^{2})> 0$, $S(\rho_2^{2}) > 0$, and $S(\rho_3^{2}) > 0$.

\section{A genuine steering inequality for three settings per site}\label{gensteer} 
\begin{theorem}\label{thmappen}
If any given quantum correlation violates the steering inequality:
\begin{eqnarray}\label{1}
|\langle (A_0 B_0+A_1B_1&+& A_2B_2)C_0\rangle+\langle (A_0 B_2-A_1B_0+A_2B_1)C_1\rangle\nonumber\\
&+&\langle (A_0 B_1-A_1B_2+A_2B_0)C_2\rangle|\leq 3,
\end{eqnarray}
then the correlation exhibits genuine tripartite steering from Charlie to Alice and Bob. Here measurements of Alice and Bob demonstrate EPR steering without Bell nonlocality  while  measurements of Charlie are uncharacterized.
\end{theorem} 
Before proving the theorem, we first prove  the following lemma:\\
\begin{lemma} 
Any LHS-LHS model satisfies the following inequality:
\begin{equation}\label{2}
    |\langle A_0 B_0+A_1B_1+A_2B_2\rangle|\leq 1.
\end{equation} 
\end{lemma}
\begin{proof}
Let us denote
\begin{equation}\label{2i}
    S_3=\langle A_0 B_0+A_1B_1+A_2B_2\rangle
\end{equation}
For any separable state (due to linearity of the quantity $S_3$, without loss of generality one can consider product states $\rho_{AB}=\rho_A\bigotimes\rho_B$ only for this purpose),
\begin{equation}\label{2ii}
    |S_3|\leq |\overrightarrow{v_{A}}.\overrightarrow{v_{B}}|,
\end{equation}
where $\overrightarrow{v_{A/B}}=(\langle A_0/B_0\rangle,\langle A_1/B_1\rangle,\langle A_2/B_2\rangle)$. By Cauchy Schwarz inequality, we get,
\begin{equation}\label{2ii}
    |S_3|\leq |\overrightarrow{v_{A}}||\overrightarrow{v_{B}}|,
\end{equation}
Now,
\begin{equation}\label{2iii}
    |\overrightarrow{v_{A}}|=\sqrt{\sum_{i=0}^2\langle A_i\rangle^2}
\end{equation}
Again,
\begin{equation}\label{2iv}
  \langle A_i\rangle=\textmd{Tr}(A_i \rho_A),
\end{equation}
where $\rho_A=\textmd{Tr}_B(\rho_{AB})$. After simple calculation, we get,
\begin{equation}\label{2v}
      \langle A_i\rangle=\overrightarrow{n_i}.\overrightarrow{r},
\end{equation}
where $\overrightarrow{r}$ denotes the Bloch vector corresponding to the state $\rho_A$ and $\overrightarrow{n_i}$ characterizes the measurement setting $A_i=\overrightarrow{n_i}.\overrightarrow{\sigma}.$  Using this relation(Eq.(\ref{2v})), Eq.(\ref{2iii}) gets simplified,
\begin{equation}\label{2vi}
    |\overrightarrow{v_{A}}|=\sqrt{\sum_{i=0}^2(\overrightarrow{n_i}.\overrightarrow{r})^2}.
\end{equation}
This on simplification becomes,
\begin{equation}\label{2vi}
    |\overrightarrow{v_{A}}|=|\overrightarrow{r}|\leq 1.
\end{equation}
Similarly it can be shown that $|\overrightarrow{v_{B}}|\leq 1.$ Hence Eq.(\ref{2ii}) becomes,
\begin{equation}\label{2vii}
    |S_3|\leq 1.
\end{equation}
\end{proof} 
\begin{proof} \textit{of Theorem}: 
Before we start with the proof we introduce the following notations:
\begin{equation}\label{m1}
    D_0=\langle A_0 B_0+A_1B_1+A_2B_2\rangle,
\end{equation}
where $A_i=\overrightarrow{v_i^A}.\overrightarrow{\sigma},$ $B_i=\overrightarrow{v_i^B}.\overrightarrow{\sigma}$ and $C_i=\overrightarrow{v_i^C}.\overrightarrow{\sigma},$  with $\overrightarrow{\sigma}=(\sigma_x,\sigma_y,\sigma_z)$ denote the Pauli observables.
\begin{equation}\label{m2}
  D_1= \langle A_0 B_2-A_1B_0+A_2B_1\rangle
\end{equation}
\begin{equation}\label{m3}
 D_2=\langle A_0 B_1-A_1B_2+A_2B_0\rangle
\end{equation}
The last two expressions $ D_1$(Eq.(\ref{m2})) and $ D_2$(Eq.(\ref{m3})) can be obtained from $D_0$(Eq.(\ref{m1})) under some specific relabeling of inputs and outputs:
\begin{itemize}
  \item For  $D_1$: $a\rightarrow a\bigoplus_2 x$, where $a\in\{0,1\}$ and $x\in\{0,1,2\}$ and $y \rightarrow y\bigoplus_3 2$ where $y\in\{0,1,2\}$. $\bigoplus_j$ denotes addition modulo $j$ for any positive integer $j$.
  \item For  $D_2$: $a\rightarrow a\bigoplus_2 x$ and $y \rightarrow y\bigoplus_3 1$.
\end{itemize}
With these notations, Eq.(\ref{1}) becomes modified as:
\begin{equation}\label{t}
     |\sum_{i=0}^2 \langle D_iC_i\rangle|\leq 3.
\end{equation}
Now, as Alice, Bob and Charlie are not in the same lab, Alice and Bob do not know which version of the game to play. So they play the average game $\sum_{i=0}^2 D_i$. Now there are two possible cases:
\begin{itemize}
  \item Alice and Bob share a separable state
  \item Alice and Bob share a EPR-steerable state
\end{itemize}
\begin{itemize} 
\item $\textit{Case 1:}$ Let correlations of Alice and Bob admit a $LHS-LHS$ model, i.e. they share a separable state. Then, by the lemma, we get, $D_i\leq 1$, $\forall i\in\{0,1,2\}$. Hence Eq.(\ref{t}) is satisfied.
\item $\textit{Case 2:}$ Now consider the case where the correlations do not admit a $LHS-LHV$ model(i.e. if the state is EPR-steerable). By quantum predictions, the algebraic maximum of the game is $3.$ For instance if Alice and Bob share the entangled state $|\psi^+\rangle$, then for a particular measurement settings $D_0=3$ whereas both of $D_1$ and $D_2=0$. Hence the theorem.
\end{itemize}  
\end{proof}




\begin{thebibliography}{99}
	
\bibitem{Reviw Ent} R. Horodecki, P. Horodecki, M. Horodecki, and K. Horodecki, "Quantum entanglement", \href{http://journals.aps.org/rmp/abstract/10.1103/RevModPhys.81.865}{Rev. Mod. Phys. {\bf 81}, 865 (2009)}

\bibitem{Wiseman07} H. M. Wiseman, S. J. Jones, and A. C. Doherty, ``Steering, Entanglement, Nonlocality, and the Einstein-Podolsky-Rosen Paradox'', \href{http://journals.aps.org/prl/abstract/10.1103/PhysRevLett.98.140402}{Phys. Rev. Lett. 98, 140402}, S. J. Jones, H. M. Wiseman, and A. C. Doherty, ``Entanglement, Einstein-Podolsky-Rosen correlations, Bell nonlocality, and steering'', \href{http://journals.aps.org/pra/abstract/10.1103/PhysRevA.76.052116}{Phys. Rev. A 76, 052116 (2007)}

\bibitem{Reviw NL} N. Brunner, D. Cavalcanti, S. Pironio, V. Scarani, and S.Wehner, "Bell nonlocality",
\href{http://journals.aps.org/rmp/abstract/10.1103/RevModPhys.86.419}{Rev. Mod. Phys. {\bf 86}, 839 (2014)}. 	
	
\bibitem{Bell64} J. S. Bell, "On the Einstein Podolsky Rosen Paradox", Physics {\bf 1} (3): 195–200 (1964), J. S. Bell, Speakable and Unspeakable in Quantum Mechanics (Cambridge University Press, 1987).

\bibitem{Bennett93} C. H. Bennett, G. Brassard, C. Cr\'{e}peau, R. Jozsa, A. Peres, and W. K. Wootters, "Teleporting an unknown quantum state via dual classical and Einstein-Podolsky-Rosen channels", \href{http://dx.doi.org/10.1103/PhysRevLett.70.1895}{Phys. Rev. Lett. {\bf 70}, 1895 (1993)}. 

\bibitem{Bennett92}  C. H. Bennett, S. J. Wiesner, "Communication via one- and two-particle operators on Einstein-Podolsky-Rosen states", \href{http://dx.doi.org/10.1103/PhysRevLett.69.2881}{Phys. Rev. Lett. 69, 2881 (1992)}. 

\bibitem{random} S. Pironio, A. Ac\'{i}n, S. Massar, A. Boyer de la Giroday, D. N. Matsukevich, P. Maunz, S. Olmschenk, D. Hayes, L. Luo, T. A. Manning and C. Monroe, " Random numbers certified by Bell’s theorem", \href{http://www.nature.com/nature/journal/v464/n7291/full/nature09008.html}{Nature {\bf 464}, 1021-1024}. 
R. Colbeck and R. Renner, ``Free randomness can be amplified", \href{http://www.nature.com/nphys/journal/v8/n6/full/nphys2300.html}{Nat. Phys.{\bf 8}, 450 (2012)};
A. Chaturvedi and M. Banik, ``Measurement-device–independent randomness from local entangled states",
\href{http://iopscience.iop.org/article/10.1209/
	0295-5075/112/30003/meta;jsessionid=D6C96ABB3E61
	C42C542A9553E8A4F4DC.c3.iopscience.cld.iop.org}{EPL {\bf 112}, 30003 (2015)}.

\bibitem{key} J. Barrett, L. Hardy, and A. Kent, ``No signaling and quantum key distribution", \href{http://dx.doi.org/10.1103/PhysRevLett.95.010503}{Phys. Rev. Lett. {\bf 95}, 010503 (2005)};
 A. Ac\'{i}n, N. Gisin, and L. Masanes, ``From Bells theorem to secure quantum key distribution", \href{http://dx.doi.org/10.1103/PhysRevLett.97.120405}{Phys. Rev. Lett. {\bf 97}, 120405 (2006)};
 
\bibitem{dw} N. Brunner, S. Pironio, A. Ac\'{i}n, N. Gisin, A. A. Methot, and	V. Scarani, ``Testing the dimension of Hilbert spaces", \href{http://dx.doi.org/10.1103/PhysRevLett.100.210503}{Phys. Rev. Lett. {\bf 100}, 210503 (2008)};
R. Gallego, N. Brunner, C. Hadley, and A. Ac\'{i}n, ``Device independent tests of classical and quantum dimensions",  \href{http://dx.doi.org/10.1103/PhysRevLett.105.230501}{Phys. Rev. Lett. {\bf 105}, 230501 (2010)};
S. Das, M. Banik, A. Rai, MD R. Gazi, and S.Kunkri, ``Hardy's nonlocality argument as a witness for postquantum correlations",
\href{https://journals.aps.org/pra/abstract/10.1103/PhysRevA.87.012112}{Phys. Rev. A {\bf 87}, 012112 (2013)};
A. Mukherjee, A. Roy, S. S. Bhattacharya, S. Das, Md. R. Gazi, and M. Banik, ``Hardy's test as a device-independent dimension witness",
\href{https://journals.aps.org/pra/abstract/10.1103/PhysRevA.92.022302}{Phys. Rev. A {\bf 92}, 022302 (2015)};

\bibitem{game} N. Brunner and N. Linden, ``Connection between Bell nonlocality and Bayesian game theory", \href{http://www.nature.com/ncomms/2013/130703/ncomms3057/full/ncomms3057.html#references}{Nature Communications {\bf 4}, 2057 (2013)}.
A. Pappa \emph{et al.} ``Nonlocality and Conflicting Interest Games", 
\href{http://journals.aps.org/prl/abstract/10.1103/PhysRevLett.114.020401}{Phys. Rev. Lett. {\bf 114}, 020401 (2015)}.
A. Roy, A. Mukherjee, T. Guha, S. Ghosh, S. S. Bhattacharya, M. Banik, ``Nonlocal correlations: Fair and Unfair Strategies in Bayesian Game'', \href{http://arxiv.org/abs/1601.02349}{Arxiv: 1601.02349}.



\bibitem{Branciard12} C. Branciard, E. G. Cavalcanti, S. P. Walborn, V. Scarani, and H. M. Wiseman, ``One-sided device-independent quantum key distribution: Security, feasibility, and the connection with steering'', \href{http://dx.doi.org/10.1103/PhysRevA.85.010301}{Phys. Rev. A {\bf 85}, 010301(R) (2012)}.

\bibitem{Quintino15}  M. T. Quintino, T. V\'{e}rtesi, D. Cavalcanti, R. Augusiak, M. Demianowicz, A. Ac\'{i}n, and N. Brunner, ``Inequivalence of entanglement, steering, and Bell nonlocality for general measurements'', \href{http://dx.doi.org/10.1103/PhysRevA.92.032107}{Phys. Rev. A {\bf 92}, 032107 (2015)}.

\bibitem{gme} Otfried G\"{u}hne, Ge\'{z}a T\'{o}th, ``Entanglement detection'', \href{http://www.sciencedirect.com/science/article/pii/S0370157309000623}{Physics Reports 474, 1}.

\bibitem{comp} R. Raussendorf and H. J. Briegel, Phys. Rev. Lett. 86, 5188 (2001).

\bibitem{sim} S. Lloyd,Science 273, 1073 (1996).

\bibitem{metr} L. Pezze and A. Smerzi,Phys. Rev. Lett. 102,
100401(2009)

\bibitem{Ramij15}  R. Rahaman, G. Kar, ``GHZ correlation provides secure Anonymous Veto Protocol'', \href{https://arxiv.org/abs/1507.00592}{Arxiv: 1507.00592}.


\bibitem{multisteer} Q. Y. He and M. D. Reid, ``Genuine Multipartite Einstein-Podolsky-Rosen Steering'', \href{http://journals.aps.org/prl/abstract/10.1103/PhysRevLett.111.250403}{Phys. Rev. Lett. 111, 250403 (2013)},
D. Cavalcanti1, P. Skrzypczyk1,2, G.H. Aguilar3, R.V. Nery3, P.H. Souto Ribeiro3 and S.P. Walborn, ``Detection of entanglement in asymmetric quantum
networks and multipartite quantum steering'', \href{http://www.nature.com/nphys/journal/v11/n2/abs/nphys3202.html}{Nature Physics {\bf 11}, 167–172(2015)}

\bibitem{Raussendorf2001}  H. J. Briegel and R. Raussendorf, "Persistent Entanglement in Arrays of Interacting Particles", \href{http://dx.doi.org/10.1103/PhysRevLett.86.910}{Phys. Rev. Lett. 86, 910(2001)}.

\bibitem{Brunner12} 
N. Brunner and T. Vertesi, ``Persistency of entanglement and nonlocality in multipartite quantum systems'', \href{http://journals.aps.org/pra/abstract/10.1103/PhysRevA.86.042113}{Phys. Rev. A 86, 042113 (2012)} 

\bibitem{Vertesi16} P. Divi\'{a}nszky, R. Trencs\'{e}nyi, E. Bene, and T. V\'{e}rtesi, ``Bounding the persistency of the nonlocality of W states'', \href{http://dx.doi.org/10.1103/PhysRevA.93.042113}{Phys. Rev. A, {\bf 93}, 042113 (2016)}

\bibitem{fei} H-H. Qin, S-M. Fei, and X. Li-Jost, ``Trade-off relations of Bell violations among pairwise qubit systems'', \href{http://dx.doi.org/10.1103/PhysRevA.92.062339}{Physical Review A {\bf 92}, 062339 (2015)}

\bibitem{toner} B. Toner, F. Verstraete, ``Monogamy of Bell correlations and Tsirelson's bound'', \href{https://arxiv.org/abs/quant-ph/0611001v1}{arXiv:quant-ph/0611001v1}.

\bibitem{entdet} O. G\"{u}hne, G. T\'{o}th, ``Entanglement detection'', \href{http://www.sciencedirect.com/science/article/pii/S0370157309000623}{Phys. Reports {\bf 474} 1 (2009)}.

\bibitem{guhne} O. G\"{u}hne, and M. Seevinck, ``Separability criteria for genuine multiparticle entanglement'', \href{http://iopscience.iop.org/article/10.1088/1367-2630/12/5/053002/meta;jsessionid=53A177F02A8D33E05B44BD6C5E6104D4.c6.iopscience.cld.iop.org}{New J. Phys. {\bf 12}, 053002 (2010) }.

\bibitem{Popescu95}  S. Popescu, "Bell's Inequalities and Density Matrices: Revealing “Hidden” Nonlocality", \href{http://dx.doi.org/10.1103/PhysRevLett.74.2619}{Phys. Rev. Lett. 74, 2619 (1995)}. 


\bibitem{Werner89}  R. F. Werner, "Quantum states with Einstein-Podolsky-Rosen correlations admitting a hidden-variable model", \href{http://dx.doi.org/10.1103/PhysRevA.40.4277}{Phys. Rev. A 40, 4277 (1989)}. 


\bibitem{mohamed2006}  M. Bourennane, M. Eibl, S. Gaertner, N. Kiesel, C. Kurtsiefer, and H. Weinfurter, "Entanglement Persistency of Multiphoton Entangled States", \href{http://dx.doi.org/10.1103/PhysRevLett.96.100502}{Phys. Rev. Lett. 96, 100502 (2006)}. 

\bibitem{Bancal} J-D. Bancal, J. Barrett, N. Gisin, and S. Pironio, ``Definitions of multipartite nonlocality'', \href{http://dx.doi.org/10.1103/PhysRevA.88.014102}{Phys. Rev. A {\bf 88}, 014102 (2013)}.


\bibitem{Ver02}
F. Verstraete, J. Dehaene, B. De Moor, H. Verschelde, \href{http://dx.doi.org/10.1103/PhysRevA.65.052112}{Phys. Rev. A \textbf{65}, 052112 (2002)}.

\bibitem{Jeba16} C. Jebaratnam, ``Detecting genuine multipartite entanglement in steering scenarios'', \href{http://journals.aps.org/pra/abstract/10.1103/PhysRevA.93.052311}{Phys. Rev. A {\bf 93}, 052311 (2016)}

\bibitem{Uh00}
A. Uhlmann, \href{http://dx.doi.org/10.1103/PhysRevA.62.032307}{Phys. Rev. A {\bf 62}, 032307 (2000)}; A. Wong and N. Christensen, \href{http://dx.doi.org/10.1103/PhysRevA.63.044301}{ibid. {\bf 63}, 044301 (2001)}; S. S. Bullock and G. K. Brennen, \href{http://scitation.aip.org/content/aip/journal/jmp/45/6/10.1063/1.1723701;jsessionid=KlXulfNwo-1QF6ZkzmOWChDu.x-aip-live-02}{J. Math. Phys. {\bf 45}, 2447 (2004).}

\bibitem{deftangle} W. K. Wootters, ``Entanglement of Formation of an Arbitrary State of Two Qubits'', \href{http://dx.doi.org/10.1103/PhysRevLett.80.2245}{Phys. Rev. Lett. {\bf 80}, 2245 (1998)}

\bibitem{GBS} G. Gour, S. Bandyopadhyay and B. C. Sanders, ``Dual monogamy inequality for entanglement'', \href{http://dx.doi.org/10.1063/1.2435088}{J. Math. Phys. 48, 012108 (2007)}.



\bibitem{Ver03} F. Verstraete, J. Dehaene, and B. De Moor, ``Normal forms and entanglement measures for multipartite quantum states'', \href{http://dx.doi.org/10.1103/PhysRevA.68.012103}{Phys. Rev. A {\bf 68}, 012103 (2003)}.


\bibitem{Sliwa03} C. Śliwa, ``Symmetries of the Bell correlation inequalities'', \href{http://www.sciencedirect.com/science/article/pii/S0375960103011150}{Phys. Lett. A, {\bf 317}, 165 (2003)}.


\bibitem{Tran14} M. C. Tran, W. Laskowski and T. Paterek, ``The Werner gap in the presence of simple coloured noise'', \href{http://iopscience.iop.org/article/10.1088/1751-8113/47/42/424025/meta;jsessionid=7A312484461B1CF9FE06FA830599B074.c2.iopscience.cld.iop.org}{J. Phys. A: Math Theo {\bf 47}, 42 (2014)}

\bibitem{Jevtic15} S. Jevtic, M. J. W. Hall, M. R. Anderson, M. Zwierz, and H. M. Wiseman, ``Einstein–Podolsky–Rosen steering and the steering ellipsoid'', \href{https://www.osapublishing.org/josab/abstract.cfm?uri=josab-32-4-A40}{Journal of the Optical Society of America B, {\bf 32}, A40 (2015)}.

\bibitem{Dicke} R. H. Dicke, ``Coherence in Spontaneous Radiation Processes'', \href{http://dx.doi.org/10.1103/PhysRev.93.99}{Phys. Rev. 93, 99 (1954)}.


\end{thebibliography}
\end{document}